\documentclass[submission,copyright,creativecommons]{eptcs}
\providecommand{\event}{SOS 2011} 

\RequirePackage{lineno}

\usepackage{breakurl}

\usepackage{url}
\usepackage{color}
\usepackage{amssymb}
\usepackage{amsmath}
\usepackage{amsthm}
\usepackage{alltt}
\usepackage{verbatim}
\usepackage{stmaryrd}
\usepackage{mathtools}
\usepackage{wrapfig}
\usepackage{float}
\usepackage{wasysym}
\usepackage[all,2cell,cmtip,dvips,ps]{xy}

\usepackage{mathrsfs}

\DeclareSymbolFont{operators}{OT1}{cmr}{m}{n}
\DeclareSymbolFont{letters}{OML}{cmm}{m}{it}
\DeclareSymbolFont{symbols}{OMS}{cmsy}{m}{n}
\DeclareSymbolFont{largesymbols}{OMX}{cmex}{m}{n}

\newtheorem{theorem}{Theorem}
\newtheorem{proposition}{Proposition}
\newtheorem{lemma}{Lemma}
\newtheorem{definition}{Definition}
\newtheorem{example}{Example}
\newtheorem{remark}{Remark}
\newtheorem{corollary}{Corollary}

\definecolor{darkgreen}{RGB}{0,150,0}
\definecolor{darkblue}{RGB}{0,0,0}

\renewcommand{\fboxsep}{1.5pt}
\newcommand{\preg}{\emph{preg}}

\newcommand{\notranz}[1]{\buildrel \text{$#1$}\over \nrightarrow}
\newcommand{\height}{\textit{height}}

\newcommand{\ET}{\textnormal{$\textit{E}$}}
\newcommand{\FINTREEPRED}{\textit{FTP}}

\newcommand{\ETz}{\textnormal{$\textit{E}_{\FINTREEPRED}$}}

\newcommand{\ETdag}{\textnormal{$\textit{E}_{\FINTREEPRED}^\partial$}}

\newcommand{\xleadsto}[1]{\textnormal{\hspace{3pt}$\overset{#1}{\leadsto}$\hspace{3pt}}}

\newcommand{\calB}{\textnormal{$\mathcal B$}}
\newcommand{\calP}{\textnormal{$\mathcal P$}}
\newcommand{\calA}{\textnormal{$\mathcal A$}}
\newcommand{\calQ}{\textnormal{$\mathcal Q$}}

\newcommand{\ct}{\textnormal{$\kappa$}}

\newcommand{\pimp}{\textnormal{$\mathcal{P^I}$}}
\newcommand{\sati}{\textnormal{${\propto}$}}
\newcommand{\rlsat}{\textnormal{$\leadsto, \sati, \sigma \models\,$}}
\newcommand{\sstrans}[1]{\textnormal{$\rightarrow_{{#1}}$}}
\newcommand{\sssat}[1]{\textnormal{$\ltimes_{{#1}}$}}

\newcommand{\onedagpar}[1]{{\textnormal{$\partial_{{#1}}$}}}
\newcommand{\onedag}{\onedagpar{\calB,\calQ}}
\newcommand{\onedagnact}{\onedagpar{\emptyset,\calQ}}

\newcommand{\setrules}{\textnormal{$\cal R$}}
\newcommand{\setrulestranz}{\textnormal{$\setrules^{\calA}$}}
\newcommand{\setrulessat}{\textnormal{$\setrules^{\calP}$}}

\newcommand{\setrulesG}{\textnormal{${\cal R}_G$}}
\newcommand{\setrulestranzG}{\textnormal{$\setrules^{\calA}_G$}}
\newcommand{\setrulessatG}{\textnormal{$\setrules^{\calP}_G$}}

\newcommand{\actpar}[1]{{\textnormal{$I^+_{#1}$}}}
\newcommand{\nactpar}[1]{{\textnormal{$I^-_{#1}$}}}
\newcommand{\satpar}[1]{{\textnormal{$J^+_{#1}$}}}
\newcommand{\nsatpar}[1]{{\textnormal{$J^-_{#1}$}}}

\newcommand{\act}{{\textnormal{$I^+$}}}
\newcommand{\nact}{{\textnormal{$I^-$}}}
\newcommand{\sat}{{\textnormal{$J^+$}}}
\newcommand{\nsat}{{\textnormal{$J^-$}}}
\newcommand{\all}{{\textnormal{$L$}}}

\newcommand{\rlAxpref}{1}
\newcommand{\rlsumLpref}{2}
\newcommand{\rlsumRpref}{3}
\newcommand{\rlAxcp}{4}
\newcommand{\rlPL}{5}
\newcommand{\rlPR}{6}
\newcommand{\rlPpref}{7}

\newcommand{\rlanin}{8}
\newcommand{\rlpnin}{9}
\newcommand{\rlaipa}{10}
\newcommand{\rlaipp}{11}

\newcommand{\eqcomm}{1}
\newcommand{\eqassoc}{2}
\newcommand{\eqidem}{3}
\newcommand{\equnit}{4}
\newcommand{\eqcp}{5}

\newcommand{\eqdagdlt}{6}
\newcommand{\eqdagcpin}{7}
\newcommand{\eqdagcpnin}{8}
\newcommand{\eqdagain}{9}
\newcommand{\eqdagaino}{9.1}
\newcommand{\eqdagaint}{9.2}
\newcommand{\eqdagainth}{9.3}
\newcommand{\eqdagainf}{9.4}
\newcommand{\eqdaganin}{10}
\newcommand{\eqdagrec}{11}
\newcommand{\eqdagsum}{12}

\newcommand{\aaipld}{13}
\newcommand{\aaiprd}{14}
\newcommand{\aaipact}{15}
\newcommand{\aaipdeadf}{16}
\newcommand{\aaipdeads}{17}
\newcommand{\aaippred}{18}

\def\titlerunning{Axiomatizing GSOS with Predicates}
\def\authorrunning{L. Aceto, G. Caltais, E.-I. Goriac, A. Ingolfsdottir}

\begin{document}


\title{Axiomatizing GSOS with Predicates\thanks{The authors have been 
 been partially supported by the projects ``New Developments in
 Operational Semantics" (nr.~080039021), ``Meta-theory of Algebraic
 Process Theories" (nr.~100014021), and ``Extending and Axiomatizing Structural Operational Semantics: Theory and Tools'' (nr. 110294-0061) of the Icelandic Research Fund.
 The work on the paper was partly carried out while Luca
 Aceto and Anna Ingolfsdottir held an Abel Extraordinary Chair at
 Universidad Complutense de Madrid, Spain, supported by the NILS
 Mobility Project.}}

\author {
Luca Aceto \qquad
Georgiana Caltais \qquad
Eugen-Ioan Goriac \qquad
Anna Ingolfsdottir
\email{
[luca,gcaltais10,egoriac10,annai]@ru.is
}
\institute {
ICE-TCS, School of Computer Science, Reykjavik University, Iceland
}
}

\maketitle

\pagestyle{plain}

\begin{abstract}
In this paper, we introduce an extension of the GSOS rule format with predicates such as termination, convergence and divergence. For this format we generalize the technique proposed by Aceto, Bloom and Vaandrager for the automatic generation of ground-complete axiomatizations of bisimilarity over GSOS systems. Our procedure is implemented in a tool that receives SOS specifications as input and derives the corresponding axiomatizations automatically. This paves the way to checking strong bisimilarity over process terms by means of theorem-proving techniques.
\end{abstract}

\section{Introduction}
\label{sec:intro}

One of the greatest challenges in computer science is the development
of rigorous methods for the specification and verification of reactive
systems, {\it{i.e.}}, systems that compute by interacting with their
environment. Typical examples include embedded systems, control
programs and distributed communication protocols.  Over the last three
decades, process algebras, such as ACP~\cite{BaetenBR2010},
CCS~\cite{Mi89} and CSP~\cite{Ho85},
have been successfully used as common languages for the description of
both actual systems and their specifications. In this context,
verifying whether the implementation of a reactive system complies to
its specification reduces to proving that the corresponding
process terms are related by some notion of behavioural equivalence or
preorder~\cite{Glabbeek01}. 

One approach to proving equivalence between two terms is to exploit
the equational style of reasoning supported by process algebras. In
this approach, one obtains a (ground-)complete axiomatization of the
behavioural relation of interest and uses it to prove the equivalence
between the terms describing the specification and the implementation
by means of equational reasoning, possibly in conjunction with proof
rules to handle recursively-defined process specifications. 

Finding a ``finitely specified", (ground-)complete axiomatization of a 
behavioural equivalence over a process algebra is often a highly non-trivial
task. However, as shown in~\cite{Aceto:1994:TSR:184662.184663} in the
setting of bisimilarity~\cite{Mi89,Pa81}, this process can be
automated for process languages with an operational semantics given in
terms of rules in the GSOS format of Bloom, Istrail and
Meyer~\cite{Bloom:1995:BCT:200836.200876}. In that reference, Aceto,
Bloom and Vaandrager provided an algorithm that, given a GSOS language
as input, produces as output a ``conservative extension" of the
original language with auxiliary operators together with a finite
axiom system that is sound and ground-complete with respect to
bisimilarity (see, {\it e.g.}, \cite{DBLP:conf/concur/Aceto94,GazdaFokkink2010,DBLP:conf/concur/1994,DBLP:journals/tcs/Ulidowski00} for further results in this line of research). As the operational specification of several operators
often requires a clear distinction between successful termination and
deadlock, an extension of the above-mentioned approach to the setting of GSOS with a predicate for termination was proposed in~\cite{DBLP:journals/jlp/BaetenV04}.

In this paper we contribute to the line of the work
in~\cite{Aceto:1994:TSR:184662.184663} and
\cite{DBLP:journals/jlp/BaetenV04}. Inspired by~\cite{DBLP:journals/jlp/BaetenV04}, we introduce the {\preg} rule
format, a natural extension of the GSOS format with an arbitrary collection of predicates such
as termination, convergence and divergence. We further adapt the
theory in~\cite{Aceto:1994:TSR:184662.184663} to this setting and give a procedure for
obtaining ground-complete axiomatizations for bisimilarity over
{\preg} systems.
More specifically, we develop a general procedure that, given a
{\preg} language as input, automatically synthesizes a conservative
extension of that language and a finite axiom system that, in
conjunction with an infinitary proof rule, yields a sound and
ground-complete axiomatization of bisimilarity over the extended
language.  The work we present in this paper is based on the one
reported
in~\cite{Aceto:1994:TSR:184662.184663,DBLP:journals/jlp/BaetenV04}. However, handling more general predicates than immediate termination requires
the introduction of some novel technical ideas. In particular, the
problem of axiomatizing bisimilarity over a {\preg} language is
reduced to that of axiomatizing that relation over finite trees whose
nodes may be labelled with predicates. In order to do so, one needs to
take special care in axiomatizing negative premises in rules that may
have positive and negative premises involving predicates and
transitions.

The results of the current paper have been used for the implementation
of a Maude~\cite{DBLP:conf/maude/2007} tool \cite{pregax-calco-tools2011}
that enables the user to specify {\preg} systems in a uniform
fashion, and that automatically derives the associated
axiomatizations. The tool is available at {\footnotesize{\url{http://goriac.info/tools/preg-axiomatizer/}}}. This paves the way to checking bisimilarity
over process terms by means of theorem-proving techniques for a large
class of systems that can be expressed using {\preg} language specifications.

\paragraph{Paper structure.}
{ In Section~\ref{sec:prelim} we introduce the {\preg} rule format.
In Section~\ref{sec:fintree} we
introduce an appropriate ``core" language for expressing finite trees
with predicates. We also provide a ground-complete axiomatization for
bisimilarity over this type of trees, as our aim is to prove
the completeness of our final axiomatization by head normalizing
general {\preg} terms, and therefore by reducing the completeness problem for arbitrary languages to that for trees.

Head normalizing general {\preg} terms is not a straightforward
process. Therefore, following~\cite{Aceto:1994:TSR:184662.184663}, in
Section~\ref{sec:smooth} we introduce the notion of smooth and distinctive
operation, adapted to the current setting. These operations are
designed to ``capture the behaviour of general {\preg} operations",
and are defined by rules satisfying a series of syntactic constraints
with the purpose of enabling the construction of head normalizing axiomatizations. Such axiomatizations are based on a collection of equations that describe the interplay between smooth and distinctive operations, and the operations in the signature for finite trees.
The existence of a sound and ground-complete axiomatization characterizing
the bisimilarity of {\preg} processes is finally proven in
Section~\ref{sec:completeness}. A technical discussion on why it is important to handle predicates as first class notions, instead of encoding them by means of transition relations, is presented in
Section~\ref{sec:rationale}. In Section~\ref{sec:conclusions} we
draw some conclusions and provide pointers to future work.

\section{GSOS with predicates}
\label{sec:prelim}

In this section we present the {\preg} systems which are a generalization of GSOS \cite{Bloom:1995:BCT:200836.200876} systems.

Consider a countably infinite set $V$ of \emph{process variables} (usually denoted by $x$, $y$, $z$) and a signature $\Sigma$ consisting of a set of \emph{operations} (denoted by $f$, $g$). The set of \emph{process terms} ${\mathbb{T}}(\Sigma)$ is inductively defined as follows: each variable $x \in {\it V}$ is a term; if $f \in \Sigma$ is an operation of arity $l$, and if $S_1, \ldots, S_l$ are terms, then $f(S_1, \ldots, S_l)$ is a term. We write $T(\Sigma)$ in order to represent the set of \emph{closed process terms} (\textit{i.e.}, terms that do not contain variables), ranged over by $t, s$.
A \emph{substitution} $\sigma$ is a function of type $V \rightarrow {\mathbb T}(\Sigma)$. If the range of a substitution is included in $T(\Sigma)$, we say that it is a \emph{closed substitution}. Moreover, we write $[x \mapsto t]$ to represent a substitution that maps the variable $x$ to the term $t$.
Let $\vec{x} = x_1, \ldots, x_n$ be a sequence of pairwise distinct variables. A $\Sigma$-\emph{context} $C[\vec{x}]$ is a term in which at most the variables $\vec{x}$ appear. For instance, $f(x,f(x,c))$ is a $\Sigma$-context, if the binary operation $f$ and the constant $c$ are in $\Sigma$.

Let $\calA$ be a finite, nonempty set of \emph{actions} (denoted by $a$, $b$, $c$). A \emph{positive transition formula} is a triple $(S, a, S')$ written $S \xrightarrow{a} S'$, with the intended meaning: process $S$ performs action $a$ and becomes process $S'$. A \emph{negative transition formula} $(S, a)$ written $S \notranz{a}$, states that process $S$ cannot perform action $a$. Note that $S, S'$ may contain variables. The ``intended meaning" applies to closed process terms.

We now define {\preg} -- \emph{pr}edicate \emph{e}xtension of the \emph{G}SOS rule format.
Let $\calP$ be a finite set of \emph{predicates} (denoted by $P, Q$). A \emph{positive predicate formula} is a pair $(P, S)$, written $PS$, saying that process $S$ satisfies predicate $P$. Dually, a \emph{negative predicate formula} $\neg P\,S$ states that process $S$ does not satisfy predicate $P$.

\begin{definition}[{\preg} rule format]
\label{def:apreg}
Consider $\cal A$, a set of actions, and $\calP$, a set of predicates.
\begin{enumerate}\itemsep1pt
\setlength\itemsep{1ex}
\item A \emph{{\preg} transition rule} for an $l$-ary operation $f$ is a deduction rule of the form:

\[
\dfrac{
\begin{array}{c@{~~~}c}
  \{ x_i \xrightarrow{a_{ij}} y_{ij} \mid i \in \act, j \in \actpar{i}  \} &
  \{ P_{ij} x_i \mid i \in \sat, j \in \satpar{i} \} 
  \\
  \{ x_i \notranz{b} \hspace{7pt} \mid i \in \nact, b \in \calB_i \} &
  \{ \neg Q x_i \mid i \in \nsat, Q \in \calQ_i \}
\end{array}
}
{
f(x_1, \ldots, x_l) \xrightarrow{c} C[\vec{x}, \vec{y}]
}
\]

where
\begin{enumerate}\itemsep1pt
\item $x_1, \ldots, x_l$ and $y_{ij}$ $(i \in \act, j \in \sat)$ are pairwise distinct variables;
\item $\act,\sat,\nact,\nsat \subseteq \all = \{1,\ldots,l\}$ and each $\actpar{i}$ and $\satpar{i}$ is finite;
\item $a_{ij}, b$ and $c$ are actions in $\cal A$ (${\calB_i} \subseteq {\cal A}$); and
\item $P_{ij}$ and $Q$ are predicates in $\calP$ (${\calQ_i} \subseteq {\calP}$).
\end{enumerate}

\item A \emph{{\preg} predicate rule} for an $l$-ary operation $f$ is a deduction rule similar to the one above, with the only difference that its conclusion has the form $P(f(x_1, \ldots, x_l))$ for some $P \in \calP$.

\end{enumerate}
\end{definition}

Let $\rho$ be a {\preg} (transition or predicate) rule for $f$.
The symbol $f$ is the \emph{principal operation} of $\rho$. All the formulas above the line are \emph{antecedents} and the formula below is the \emph{consequent}. We say that a position $i$ for $\rho$ is \emph{tested positively} if $i \in \act  \cup \sat$ and $\actpar{i} \cup \satpar{i} \not= \emptyset$. Similarly, $i$ is \emph{tested negatively} if $i \in \nact  \cup \nsat$ and $\calB_i \cup \calQ_i \not= \emptyset$.
Whenever $\rho$ is a transition rule for $f$, we say that
$f(\vec{x})$ is the \emph{source}, $C[\vec{x}, \vec{y}]$ is the \emph{target}, and $c$ is the \emph{action} of $\rho$.
Whenever $\rho$ is a predicate rule for $f$, we call $f(\vec{x})$ the \emph{test} of $\rho$.

In order to avoid confusion, if in a certain context we use more than one rule, e.g. $\rho, \rho'$, we parameterize the corresponding sets of indices with the name of the rule, {\it e.g.},
$\actpar{\rho}$, $\nsatpar{\rho'}$.

\begin{definition}[{\preg} system]
\label{def:apregSys}
A {\preg} system is a pair $G = (\Sigma_G, \setrulesG)$, where $\Sigma_G$ is a finite signature and $\setrulesG = \setrulestranzG \cup \setrulessatG$ is a finite set of {\preg} rules over $\Sigma_G$
({\setrulestranzG} and {\setrulessatG} represent the transition and, respectively, the predicate rules of $G$).
\end{definition}


Consider a {\preg} system $G$.
Formally, the operational semantics of the closed process terms in $G$ is fully characterized by the relations
$\sstrans{G} \subseteq T(\Sigma_G) \times {\calA} \times T(\Sigma_G)$ and $\sssat{G} \subseteq \calP \times T(\Sigma_G)$, called the (unique) \emph{sound and supported} transition and, respectively, predicate relations.
Intuitively, soundness guarantees that $\sstrans{G}$ and $\sssat{G}$ are closed with respect to the application of the rules in $\setrules_G$ on $T(\Sigma_G)$, {\it i.e.}, $\sstrans{G}$ (resp. $\sssat{G}$) contains the set of all possible transitions (resp. predicates) process terms in $T(\Sigma_G)$ can perform (resp. satisfy) according to $\setrules_G$.
The requirement that $\sstrans{G}$ and $\sssat{G}$ be supported means that all the transitions performed (resp. all the predicates satisfied) by a certain process term can be ``derived" from the deductive system described by $\setrules_G$.
As a notational convention, we write $S \xrightarrow{a}_{G} S'$ and $P_{G} S$ whenever $(S,a,S') \in\, \rightarrow_{G}$ and $(P,S) \in \sssat{G}$. We omit the subscript $G$ when it is clear from the context.


\begin{lemma}
\label{lm:fin-bran}
Let $G$ be a {\preg} system. Then, for each $t \in T(\Sigma_G)$ the set $\{(a, t') \mid t \xrightarrow{a} t',\, a\in \calA\}$ is finite.
\end{lemma}

Next we introduce the notion of \emph{bisimilarity} -- the equivalence over processes we consider in this paper.

\begin{definition}[Bisimulation]
\label{def:bisimulation}
Consider a {\preg} system $G=(\Sigma_G, \setrulesG)$.
A symmetric relation $R\, \subseteq T(\Sigma_G) \times T(\Sigma_G)$ is a \emph{bisimulation} iff:
\begin{enumerate}\itemsep1pt
\item for all $s, t, s' \in T(\Sigma_G)$, whenever $(s, t) \in\, R$ and $s \xrightarrow{a} s'$ for some $a \in {\cal A}$, then there is some $t' \in T(\Sigma_G)$ such that $t \xrightarrow{a} t'$ and $(s',t') \in\, R$;\label{def:bisim1}
\item whenever $(s,t) \in\, R$ and $Ps$ $(P \in \calP)$ then $Pt$. \label{def:bisim2}
\end{enumerate}

Two closed terms $s$ and $t$ are \emph{bisimilar} (written $s \sim t$) iff there is a bisimulation relation $R$ such that $(s,t) \in\, R$.
\end{definition}

\begin{proposition}
\label{bis:equiv-congr}
Let $G$ be a {\preg} system. Then $\sim$ is an equivalence relation and a congruence for all operations $f$ of $G$.
\end{proposition}

\begin{definition}[Disjoint extension]
\label{def:disjExt}
A {\preg} system $G'$ is a disjoint extension of a {\preg} system $G$, written $G \sqsubseteq G'$, if the signature and the rules of $G'$ include those of $G$, and $G'$ does not introduce new rules for operations in $G$.
\end{definition}

It is well known that if $G \sqsubseteq G'$ then two terms in $T(\Sigma_G)$ are bisimilar in $G$ if and only if they are bisimilar in $G'$.

{
From this point onward, our focus is to find a \emph{sound and ground-complete axiomatization of bisimilarity on closed terms} for an arbitrary {\preg} system $G$, {\it i.e.}, to identify a (finite) axiom system $\ET_G$ so that
$
\ET_G \vdash s = t\,\, \it{iff}\,\, s \sim t \text{ for all } s, t \in T$$(\Sigma_G).
$
}
The method we apply is an adaptation of the technique in~\cite{Aceto:1994:TSR:184662.184663} to the {\preg} setting.
The strategy is to incrementally build a finite, head-normalizing axiomatization  for general {\preg} terms, {\it{i.e.}}, an axiomatization that, when applied recursively, reduces the completeness problem for arbitrary terms to that for synchronization trees. This way, the proof of ground-completeness for $G$ reduces to showing the equality of closed tree terms.

\section{Preliminary steps towards the axiomatization}
\label{sec:fintree}

In this section we start by identifying an appropriate language for expressing finite trees with predicates. We continue in the style of \cite{Aceto:1994:TSR:184662.184663}, by extending the language with a kind of restriction operator used for expressing the inability of a process to perform a certain action or to satisfy a given predicate. (This operator is used in the axiomatization of negative premises.) We provide the structural operational semantics of the resulting language, together with a sound and ground-complete axiomatization of bisimilarity on finite trees with predicates.

\subsection{Finite trees with predicates}

The language for trees we use in this paper is an extension with predicates of the language BCCSP~\cite{Glabbeek01}. The syntax of BCCSP consists of closed terms built from 
a constant $\delta$ (\emph{deadlock}), the binary operator $\_\hspace{-1.5pt}+\hspace{-2.2pt}\_$\, (\emph{nondeterministic choice}), and the unary operators $a.\_$ (\emph{action prefix}), where $a$ ranges over the actions in a set {\calA}. Let {\calP} be a set of predicates. For each $P \in \calP$ we consider a process constant $\ct_{P}$, which ``witnesses'' the associated predicate in the definition of a process. Intuitively, $\ct_P$ stands for a process that only satisfies predicate $P$ and has no transition.

A finite tree term $t$ is built according to the following grammar:
\begin{equation}
\label{treesgrammar}
t {\,\,::=\,\,} \delta \mid \ct_P ~ (\forall P \in \calP) \mid
           a.t ~ (\forall a \in \calA) \mid t + t.
\end{equation}


Intuitively, $\delta$ represents a process that does not exhibit any behaviour, $s + t$ is the nondeterministic choice between the behaviours of $s$ and $t$, while $a.t$ is a process that first performs action $a$ and behaves like $t$ afterwards. The operational semantics that captures this intuition is given by the rules of BCCSP:

\begin{figure}[H]
\begin{center}
\begin{tabular}{c@{\hspace{6ex}}c@{\hspace{6ex}}c}
$\dfrac{}{a.x \xrightarrow{a} x}$ $(rl_\rlAxpref)$ &
$\dfrac{x \xrightarrow{a} x'}{x + y \xrightarrow{a} x'}$ $(rl_\rlsumLpref)$ &
$\dfrac{y \xrightarrow{a} y'}{x + y \xrightarrow{a} y'}$ $(rl_\rlsumRpref)$
\end{tabular}
\caption{The semantics of BCCSP}
\label{fig:BCCSP}
\end{center}
\end{figure}

As our goal is to extend BCCSP, the next step is to find an appropriate semantics for predicates. As can be seen in Fig.~\ref{fig:BCCSP}, action performance is determined by the shape of the terms.
Consequently, we choose to define predicates in a similar fashion.

Consider a predicate $P$ and the term $t=\ct_P$. As previously mentioned, the purpose of $\ct_P$ is to witness the satisfiability of $P$. Therefore, it is natural to consider that $\ct_P$ satisfies $P$.

Take for example the \emph{immediate termination} predicate $\downarrow$. As a term $s+s'$ exhibits the behaviour of both $s$ and $s'$, it is reasonable to state that $(s + s')\downarrow$ if $s\downarrow$ or $s'\downarrow$. Note that for a term $t=a.t'$ the statement $t\downarrow$ is in contradiction with the meaning of immediate termination, since $t$ can initially only execute action $a$. Predicates of this kind are called \emph{explicit predicates} in what follows.

Consider now the \emph{eventual termination} predicate \lightning. In this situation, it is proper to consider that $(s + t)\textnormal{\lightning}$ if $s\textnormal{\lightning}$ or $t\textnormal{\lightning}$ and, moreover, that
$a.s\textnormal{\lightning}$ if $s\textnormal{\lightning}$. We refer to predicates such as $\textnormal{\lightning}$ as \emph{implicit predicates} (that range over a set {\pimp} included in \calP), since their satisfiability propagates through the structure of tree terms in an implicit fashion. We denote by $\calA_{P}$ (included in $\calA$) the set consisting of the actions $a$ for which this behaviour is permitted when reasoning on the satisfiability of predicate $P$.

The rules expressing the semantics of predicates are:

\begin{figure}[H]
\begin{center}
\begin{tabular}{c@{\hspace{3ex}}c@{\hspace{3ex}}c@{\hspace{3ex}}c}
$\dfrac{}{P\ct_P}$ $(rl_\rlAxcp)$ & $\dfrac{Px}{P(x + y)}$ $(rl_\rlPL)$ & $\dfrac{Py}{P(x+y)}$ $(rl_\rlPR)$ &  $\dfrac{Px}{P(a.x)}, \forall P \in \pimp \,\, \forall a \in \calA_{P} ~(rl_\rlPpref)$ \\[4ex]
\end{tabular}
\caption{The semantics of predicates}
\label{fig:TTSz}
\end{center}
\end{figure}

The operational semantics of trees with predicates is given by the set of rules ($rl_\rlAxpref$)--($rl_\rlPpref$) illustrated in Fig.~\ref{fig:BCCSP} and Fig.~\ref{fig:TTSz}.
For notational consistency, we make the following conventions.
Let $\calA$ be an action set and $\calP$ a set of predicates.
$\Sigma_{\FINTREEPRED}$ represents the signature of finite trees with predicates.
$T(\Sigma_{\FINTREEPRED})$ is the set of (closed) tree terms built over $\Sigma_{\FINTREEPRED}$, and
${\setrules}_{\FINTREEPRED}$ is the set of rules ($rl_\rlAxpref$)--($rl_\rlPpref$). Moreover, by $\FINTREEPRED$ we denote the system $(\Sigma_{\FINTREEPRED}, {\setrules}_{\FINTREEPRED})$.

\paragraph{Discussion on the design decisions.}
{
At first sight, it seems reasonable for our framework to allow for language specifications containing rules of the shape $\frac{}{P(x + y)}$, or just one of ($rl_\rlPL$) and ($rl_\rlPR$). We decided, however, to disallow them, as their presence would invalidate standard algebraic properties such as the idempotence and the commutativity of $\_\hspace{-1.5pt}+\hspace{-2pt}\_\,$.

Without loss of generality we avoid rules of the form $ \frac{}{P(a.x)} $. As far as the user is concerned, in order to express that $a.x$ satisfies a predicate $P$, one can always add the witness $\ct_P$ as a summand: $a.x + \ct_P$. This decision helped us avoid some technical problems for the soundness and completeness proofs for the case of the restriction operator $\onedag$, which is presented in Section~\ref{subsec:dagger}. 

Due to the aforementioned restriction, we also had to leave out universal predicates with rules of the form $\frac{P x ~ P y}{P(x + y)}$. However, the elimination of universal predicates is not a theoretical limitation to what one can express, since a universal predicate can always be defined as the negation of an existential one.

As a last approach, we thought of allowing the user to specify existential predicates using rules of the form $\frac{P_1 x \ldots P_n x}{P(x+y)} (*)$ and $\frac{P_1 y \ldots P_n y}{P(x+y)} (**)$ (instead of $(rl_\rlPL)$ and $(rl_\rlPR)$). However, in order to maintain the validity of the axiom $x + x = x$ in the presence of rules of these forms, it would have to be the case that one of the predicates $P_i$ in the premises is $P$ itself. (If that were not the case, then let $t$ be the sum of the constants witnessing the $P_i$'s  for a rule of the form $(*)$ above with a minimal set of set premises. We have that $t + t$ satisfies $P$ by rule $(*)$. On the other hand, $P t$ does not hold since none of the $P_i$ is equal to $P$ and no rule for $P$ with a smaller set of premises exists.) Now, if a rule of the form $(*)$ has a premise of the form $P x$, then it is subsumed by $(rl_{\rlPL})$ which we must have to ensure the validity of laws such as $\ct_P = \ct_P + \ct_P$.

}


\subsection{Axiomatizing finite trees}
\label{sec:axFinTrees}

In what follows we provide a finite sound and ground-complete axiomatization ({\ETz}) for bisimilarity over finite trees with predicates.

The axiom system {\ETz} consists of the
following axioms:

\begin{figure}[H]
\begin{center}
\begin{tabular}{r@{\hspace{3pt}}c@{\hspace{3pt}}lr@{\hspace{20pt}}r@{\hspace{3pt}}c@{\hspace{3pt}}lr}
$x + y$ & = & $y + x$ & $(A_\eqcomm)$ &
$x + x$ & = & $x$ & $(A_\eqidem)$
\\[1ex]
$(x + y) + z$ & = & $x + (y + z)$ & $(A_\eqassoc)$ &
$x + \delta$ & = & $x$ & $(A_\equnit)$
\\[1ex]
\multicolumn{8}{c}{{$a.(x+\ct_P)$  =  $a.(x+\ct_P) + \ct_P, \forall P \in \pimp \,\, \forall a \in \calA_{P}$  $(A_\eqcp)$}}
\end{tabular}
\caption{The axiom system \ETz}
\label{fig:ETz}
\end{center}
\end{figure}

Axioms $(A_\eqcomm)$--$(A_\equnit)$ are well-known %
\cite{Mi89}.
Axiom $(A_\eqcp)$ describes the propagation of witness constants for the case of implicit predicates.

We now introduce the notion of terms in \emph{head normal form}. This concept plays a key role in the proofs of completeness for the axiom systems generated by our framework.

\begin{definition}[{Head} Normal Form]
Let $\Sigma$ be a signature such that $\Sigma_{\FINTREEPRED} \subseteq \Sigma$.
A term $t$ in $T(\Sigma)$ is in \emph{head normal form} (for short, h.n.f.) if
\label{def:hnf}
~\\
\[
t = {\sum_{i \in I}a_i.t_i + \sum_{j \in J}\ct_{P_j}}, \text{and the $P_j$ are all the predicates satisfied by $t$}.
\]

The empty sum $(I = \emptyset, J = \emptyset)$ is denoted by the deadlock constant $\delta$.

\end{definition}

\begin{lemma}
\label{lem:hnf_exists}
{\ETz} is head normalizing for terms in $T(\Sigma_{\FINTREEPRED})$.
That is, for all $t$ in $T(\Sigma_{\FINTREEPRED})$, there exists $t'$ in $T(\Sigma_{\FINTREEPRED})$ in h.n.f. such that $\ETz \vdash t = t'$ holds.
\end{lemma}

\begin{proof}
{The reasoning is by induction on the structure of $t$.}
\end{proof}

\begin{theorem}
\label{thm:soundness_completeness}
$\ETz$ is sound and ground-complete for bisimilarity on $T(\Sigma_{\FINTREEPRED})$. That is, it holds that
$(\forall t, t' \in T(\Sigma_{\FINTREEPRED}))\,.\, \ETz \vdash t = t' \text{ iff } t \sim t.
$ 
\end{theorem}

\subsection{Axiomatizing negative premises}
\label{subsec:dagger}

A crucial step in finding a complete axiomatization for {\preg} systems is the ``axiomatization'' of negative premises (of the shape $x \notranz{a},\, \neg P x$). In the style of \cite{Aceto:1994:TSR:184662.184663}, we introduce the restriction operator $\onedag$, where $\calB \subseteq \calA$ and $\calQ \subseteq \calP$ are the sets of initially forbidden actions and predicates, respectively. The semantics of $\onedag$ is given by the two types of transition rules in Fig.~\ref{fig:TSSdag}.

\begin{figure}[H]
\begin{center}
\begin{tabular}{c@{~~~}c}
$\dfrac{x \xrightarrow{a} x'}{\onedag(x) \xrightarrow{a} {\onedagpar{\emptyset, \calQ \cap \pimp }(x')}} ~ \text{ if } a \not \in \calB ~ (rl_{\rlanin})$
&
$\dfrac{Px}{P(\onedag(x))} ~ \text{ if } P \not \in \calQ ~ (rl_{\rlpnin})$ 
\end{tabular}
\caption{The semantics of $\onedag$}
\label{fig:TSSdag}
\end{center}
\end{figure}

Note that $\onedag$ behaves like the one step restriction operator in \cite{Aceto:1994:TSR:184662.184663} for the actions in $\cal B$, as the restriction on the action set disappears after one transition. On the other hand, for the case of predicates in $\cal Q$, the operator $\onedag$ resembles the CCS restriction operator~\cite{Mi89} since, due to the presence of implicit predicates, not all the restrictions related to predicate satisfaction necessarily disappear after one step,
as will become clear in what follows.

We write $\ET_{\FINTREEPRED}^{\partial}$ for the extension of $\ET_{\FINTREEPRED}$ with the axioms involving $\onedag$ presented in Fig.~\ref{fig:ETdag}. ${\setrules}_{\FINTREEPRED}^{\partial}$ stands for the set of rules $(rl_\rlAxpref)\!\!-\!\!(rl_{\rlpnin})$, while $\FINTREEPRED^{\partial}$ represents the system $(\Sigma_{\FINTREEPRED}^{\partial}, {\setrules}_{\FINTREEPRED}^{\partial})$.

\begin{figure}[H]
\begin{center}

\begin{tabular}{r@{\hspace{3pt}}c@{\hspace{3pt}}lll@{\hspace{20pt}}r@{\hspace{3pt}}c@{\hspace{3pt}}lll}
$\onedag(\delta)$ & = & $\delta$ & & $(A_{\eqdagdlt})$ & $\onedag(a.x)$ & = & $\sum_{{P \not \in \calQ, P(a.x)}}  \ct_P$ & if $a \in \calB$ & $(A_{\eqdagain})$\\[1ex]
$\onedag(\ct_P)$ & = & $\delta$  & if $P \in \calQ$ & $(A_{\eqdagcpin})$ & $\onedag(a.x)$ & = & $\onedagnact(a.x)$ & if $a \not \in \calB$ & $(A_{\eqdaganin})$\\[1ex]
$\onedag(\ct_P)$ & = & $\ct_P$  & if $P \not \in  \calQ$ & $(A_{\eqdagcpnin})$ & $\onedagnact(a.x)$ & = & 
{$a.\onedagpar{\emptyset, \calQ \cap \pimp }(x)$} &
& $(A_{\eqdagrec})$\\[1ex]
\multicolumn{10}{c}{{$\onedag(x+y)$ = $\onedag(x) + \onedag(y)$ $(A_{\eqdagsum})$}}
\end{tabular}
\caption{The axiom system $\ET_{\FINTREEPRED}^{\partial} \setminus \ETz$}
\label{fig:ETdag}

\end{center}
\end{figure}

Axiom $(A_{\eqdagdlt})$ states that it is useless to impose restrictions on $\delta$, as $\delta$ does not exhibit any behaviour.
The intuition behind $(A_{\eqdagcpin})$ is that since a predicate witness $\ct_P$ does not perform any action, inhibiting the satisfiability of $P$ leads to a process with no behaviour, namely $\delta$. Consequently, if the restricted predicates do not include $P$, the resulting process is $\ct_P$ itself (see $(A_{\eqdagcpnin})$).
Inhibiting the only action a process $a.t$ can perform leads to a new process that, in the best case, satisfies some of the predicates in $\pimp$ satisfied by $t$ (by $(rl_\rlPpref)$) if $\calQ \not= \pimp$ (see $(A_{\eqdagain})$).
Whenever the restricted action set $\calB$ does not contain the only action a process $a.t$ can perform, then it is safe to give up $\calB$ (see $(A_{\eqdaganin})$).
As a process $a.t$ only satisfies the predicates also satisfied by $t$, it is straightforward to see that $\onedagpar{\emptyset, \calQ}(a.t)$ is equivalent to the process obtained by propagating the restrictions on implicit predicates deeper into the behaviour of $t$ (see $(A_{\eqdagrec})$).
Axiom $(A_{\eqdagsum})$ is given in conformity with the semantics of 
$\_\hspace{-1.5pt}+\hspace{-2pt}\_\,$ ($s+t$ encapsulates both the behaviours of $s$ and $t$).

\begin{remark}
For the sake of brevity and readability, in Fig.~\ref{fig:ETdag} we presented $(A_{\eqdagain})$, which is a schema with infinitely many instances. However, it can be replaced by a finite family of axioms. See Appendix~D in the full version of the paper available at \textnormal{\footnotesize{\url{http://www.ru.is/faculty/luca/PAPERS/axgsos.pdf}}} for details.
\end{remark}

\begin{theorem}
\label{thm:soundness_hnf_dagger}
The following statements hold for $\ETdag$:
\begin{enumerate}\itemsep1pt
\item $\ETdag$ is sound for bisimilarity on $T(\Sigma_{\FINTREEPRED}^{\partial})$. \label{stmt:sound}


\item $\forall t \in T(\Sigma_{\FINTREEPRED}^{\partial}), \exists t' \in T(\Sigma_{\FINTREEPRED}) ~s.t.~ \ETdag \vdash t = t'$. \label{stmt:hnf}
\end{enumerate}

\end{theorem}

As proving completeness for $\FINTREEPRED^{\partial}$ can be reduced to showing completeness for $\FINTREEPRED$ (already proved in Theorem~\ref{thm:soundness_completeness}), the following result is an immediate consequence of Theorem~\ref{thm:soundness_hnf_dagger}:

\begin{corollary}
\label{cor:completeness_dagger}
$\ETdag$ is sound and complete for bisimilarity on $T(\Sigma_{\FINTREEPRED}^{\partial})$. 
\end{corollary}

\section{Smooth and distinctive operations}
\label{sec:smooth}

Recall that our goal is to provide a sound and ground-complete axiomatization for bisimilarity on systems specified in the {\preg} format. As the {\preg} format is too permissive for achieving this result directly, our next task is to find a class of operations for which we can build such an axiomatization by ``easily" reducing it to the completeness result for $\FINTREEPRED$, presented in Theorem~\ref{thm:soundness_completeness}. In the literature, these operations are known as \emph{smooth and distinctive}~\cite{Aceto:1994:TSR:184662.184663}. As we will see, these operations are incrementally identified by imposing suitable restrictions on {\preg} rules. The standard procedure is to first find the \emph{smooth} operations, based on which one determines the \emph{distinctive} ones.

\begin{definition}[Smooth operation]
\label{def:smooth}
~
\begin{enumerate}\itemsep1pt
\setlength\itemsep{1ex}
\item \label{eq:smoothtranz} A {\preg} transition rule is \emph{smooth} if it is of the following format:

\begin{equation}
\notag
\dfrac{
\begin{array}{c@{~~~}c}
  \{ x_i \xrightarrow{a_{i}} y_{i} \mid i \in \act \} &
  \{ P_{i} x_i \mid i \in \sat \} 
  \\
  \{ x_i \notranz{b} \hspace{7pt} \mid i \in \nact, b \in \calB_i \} &
  \{ \neg Q x_i \mid i \in \nsat, Q \in \calQ_i \}
\end{array}
}
{
f(x_1, \ldots, x_l) \xrightarrow{c}
C[\vec{x}, \vec{y}]
}
\end{equation}

where
\begin{enumerate}\itemsep1pt
\item \label{def:smooth:premise}$\act,\sat,\nact,\nsat$ disjointly cover the set $L = \{1,\ldots,l\}$,
\item in the target $C[\vec{x}, \vec{y}]$ we allow only:
$y_i$ $(i \in \act)$, $x_i$ $(i \in \nact \cup \nsat)$. \label{lb:no-x-target}
\end{enumerate}

\item \label{eq:smoothsat} A $\preg$ predicate rule is {\em smooth} if it has the form above, its premises satisfy condition (\ref{def:smooth:premise}) and its conclusion is $P(f(x_1,\ldots, x_l))$ for some $P \in\mathcal{P}$. 

\item An operation $f$ of a {\preg} system is \emph{smooth} if all its (transition and predicate) rules are smooth.

\end{enumerate}
\end{definition}

By Definition~\ref{def:smooth}, a rule $\rho$ is smooth if it satisfies the following properties:
\begin{itemize}\itemsep1pt
\item a position $i$ cannot be tested both positively and negatively at the same time,

\item positions tested positively are either from {\act} or {\sat} and they are not tested for the performance of multiple transitions (respectively, for the satisfiability of multiple predicates) within the same rule, and

\item if $\rho$ is a transition rule, then the occurrence of variables at positions $i \in \act \cup \sat$ is not allowed in the target of the consequent of $\rho$.
\end{itemize}

\begin{remark}
Note that we can always consider a position $i$ that does not occur as a premise in a rule for $f$ as being negative, with the empty set of constraints (i.e. either $i \in \nact$ and $\calB_i = \emptyset$, or $i \in \nsat$ and $\calQ_i = \emptyset$).
\end{remark}

\begin{definition}[Distinctive operation]
\label{def:smooth-distinctive}
An operation $f$ of a $\preg$ system is \emph{distinctive} if it is smooth and:
\begin{itemize}\itemsep1pt
\item for each argument $i$, either all rules for $f$ test $i$ positively, or none of them does, and

\item for any two distinct rules for $f$ there exists a position $i$ tested positively, such that one of the following holds:
\begin{itemize}\itemsep1pt
\item[-] both rules have actions that are different in the premise at position $i$,
\item[-] both rules have predicates that are different in the premise at position $i$,
\item[-] one rule has an action premise at position $i$, and the other rule has a predicate test at the same position $i$.
\end{itemize}
\end{itemize}
\end{definition}

According to the first requirement in Definition~\ref{def:smooth-distinctive}, we state that for a smooth and distinctive operation $f$, a position $i$ is \emph{positive} (respectively, \emph{negative}) for $f$ if there is a rule for $f$ such that $i$ is tested positively (respectively, negatively) for that rule.

The existence of a family of smooth and distinctive operations ``describing the behaviour" of a general {\preg} operation is formalized by the following lemma:

\begin{lemma}
\label{lm:distinctive-law-all}
Consider a {\preg} system $G$. Then there exist a {\preg} system $G'$, which is a disjoint extension of $G$ and $\FINTREEPRED$, and a finite axiom system $\ET$ such that
\begin{enumerate}
\item $\ET$ is sound for bisimilarity over any disjoint extension $G''$ of $G'$, and
\item for each term $t$ in $T(\Sigma_{G})$ there is some term $t'$ in $T(\Sigma_{G'})$ such that $t'$ is built solely using smooth and distinctive operations and $\ET$ proves $t = t'$.
\end{enumerate}

\end{lemma}


\subsection{Axiomatizing smooth and distinctive {\preg} operations}
\label{sec:axiom-amoothAndSist}



To start with, consider, for the good flow of the presentation, that we only handle explicit predicates (\emph{i.e.}, we take $\pimp = \emptyset$). Towards the end of the section we discuss how to extend the presented theory to implicit predicates.
We proceed in a similar fashion to~\cite{Aceto:1994:TSR:184662.184663} by defining a set of laws used in the construction of a complete axiomatization for bisimilarity on terms built over smooth and distinctive operations. The strength of these laws lies in their capability of reducing terms to their head normal form, thus reducing completeness for general {\preg} systems to completeness of $\ETz$ (which has already been proved in Section~\ref{sec:axFinTrees}).

\begin{definition}
\label{def:axioms}
Let $f$ be a smooth and distinctive $l$-ary operation of a {\preg} system $G$, such that ${\FINTREEPRED}^{\partial} \sqsubseteq G$.
\begin{enumerate}\itemsep1pt
\setlength\itemsep{1ex}
\item \label{enm:distr} For a positive position $i \in \all = \{1, \ldots, l\}$, the \emph{distributivity law} for $i$ w.r.t. $f$ is given as follows:
\[
f(X_1, \ldots, X'_i + X''_i, \ldots, X_l) =
f(X_1, \ldots, X'_i, \ldots, X_l) +
f(X_1, \ldots, X''_i, \ldots, X_l).
\]

\item \label{enm:actpred}
{For a rule $\rho \in \setrules$ for $f$ the \emph{{trigger law}} is, depending on whether $\rho$ is a transition or a predicate rule:}

\[
f(\vec{X}) = 
\left\{
\begin{array}{rll}
c.C[\vec{X}, \vec{y}] &,~\rho \in \setrulestranz &\emph{{(action law)}}\\
\ct_P &,~\rho \in \setrulessat &\emph{{(predicate law)}}
\end{array}
\right.
\]

\noindent
where

\[
X_i \equiv
\left\{
\begin{array}{rl}
a_i.y_i &,~i \in \act\\
\ct_{P_i} &,~i \in \sat\\
\onedagpar{{\calB_i}, {\calQ_i}}(x_i) &,~i \in \nact \cup \nsat\\
\end{array}
\right..
\]

\item \label{enm:dead}
Suppose that for $i \in \all$, term $X_i$ is in one of the forms $\delta, z_{i}, \ct_{P_i}, a.z_i, a.z_i + z'_i$ or $\ct_{P_i} + z_i$.
Suppose further that for each rule for $f$ there exists $X_j \in \vec{X}$ ($j \in \{1, \ldots, l\}$) s.t. one of the following holds:

\begin{itemize}\itemsep1pt
\item $j \in \act$ and ($X_j \equiv \delta$ or $X_j \equiv b.z_j ~ (b \not= a_j)$ or $X_j \equiv \ct_Q$, for some $Q$),
\item $j \in \sat$ and ($X_j \equiv \delta$ or $X_j \equiv \ct_{Q}~ \text{(}Q \not= P_j\text{)}$ or $X_j \equiv b.z_j$, for some $b$),
\item $j \in \nact$ and $X_j \equiv b.z_j + z'_j$, where $b \in \calB_j$,
\item $j \in \nsat$ and $X_j \equiv \ct_Q + z_j$, where $Q \in \calQ_j$.
\end{itemize}

Then the \emph{deadlock law} is as follows:
\[
f(\vec{X}) = \delta.
\]
\end{enumerate}
\end{definition}

\begin{example}
\label{ex:toy}

Consider the \emph{right-biased sequential composition} operation $\_\hspace{1.5pt};^{r}\hspace{-3pt}\_$\hspace{2pt}, whose semantics is given by the rules $
\frac{x\downarrow ~ y ~\xrightarrow{a}~ y'}{x ~;^{r}~ y ~\xrightarrow{a}~ y'}$,\hspace{0.5ex}$
\frac{x\downarrow ~ y\downarrow}{(x ~;^{r}~ y)\downarrow}$, and\hspace{0.5ex}$
\frac{x\downarrow ~ y\uparrow}{(x ~;^{r}~ y)\uparrow}
$, where $\downarrow$ \,and $\uparrow$ are, respectively, the \emph{immediate termination} and \emph{immediate divergence} predicates. 
$\_\hspace{1.5pt};^{r}\hspace{-3pt}\_$\hspace{2pt} is one of the auxiliary operations generated by the algorithm for deriving smooth and distinctive operations when axiomatizing the \emph{sequential composition} in the presence of the two mentioned predicates.

The laws derived according to Definition~\ref{def:axioms}\, for this system are:

\begin{center}
\begin{tabular}{rl@{\hspace{4ex}}rl}
$(x + y) ~;^{r} z$ &$=~~ x ~;^{r} z ~~+~~ y ~;^{r} z$&
$\delta ~;^{r} y$ &$=~~ \delta$\\

$x ~;^{r} (y + z)$ &$=~~ x ~;^{r} y ~~+~~ z ~;^{r} z$&
$k_{\uparrow} ~;^{r} y$ &$=~~ \delta$\\

$k_{\downarrow} ~;^{r} a.y$ &$=~~ a.y$&
$a.x ~;^{r} y$ &$=~~ \delta$\\

$k_{\downarrow} ~;^{r} k_{\downarrow}$ &$=~~ k_{\downarrow}$&
$x ~;^{r} \delta$ &$=~~ \delta$\\

$k_{\downarrow} ~;^{r} k_{\uparrow}$ &$=~~ k_{\uparrow}$&
\ldots

\end{tabular}
\end{center}

\end{example}

\medskip

\begin{theorem}
\label{thm:s&c-smooth-dist}
Consider $G$ a {\preg} system such that ${\FINTREEPRED}^{\partial} \sqsubseteq G$. Let $\Sigma \subseteq \Sigma_G \setminus \Sigma_{\FINTREEPRED}^{\partial}$ be a collection of smooth and distinctive operations of $G$. Let $\ET_G$ be the finite axiom system that extends $\ETdag$ with the following axioms for each $f \in \Sigma$:
\begin{itemize}\itemsep1pt
\item for each positive argument $i$ of $f$, a distributivity law (Definition~\ref{def:axioms}.\ref{enm:distr}),
\item for each transition rule for $f$, an action law (Definition~\ref{def:axioms}.\ref{enm:actpred}),
\item for each predicate rule for $f$, a predicate law (Definition~\ref{def:axioms}.\ref{enm:actpred}), and
\item all deadlock laws for $f$ (Definition~\ref{def:axioms}.\ref{enm:dead}).
\end{itemize}
The following statements hold for $\ET_G$, for any $G'$ such that $G \sqsubseteq G'$:
\begin{enumerate}\itemsep1pt
\item $\ET_G$ is sound for bisimilarity on $T(\Sigma_{G'})$. \label{stmt:sd-sound}

\item $\ET_G$ is head normalizing for $T(\Sigma \cup \Sigma_{\FINTREEPRED}^{\partial})$. \label{stmt:sd-hnf}
\end{enumerate}
\end{theorem}

Obtaining the soundness of the action law  (Definition~\ref{def:axioms}.\ref{enm:actpred}) requires some care when allowing for specifications with implicit predicates ($\pimp \not= \emptyset$).
Consider a scenario in which a transition rule for a smooth and distinctive operation $f$ is of the form $\frac{H}{f(\vec{X}) \xrightarrow{c} C[\vec{X}, \vec{y}]}$.
Assume the closed instantiation $\vec{X} = \vec{s}$, $\vec{y} = \vec{t}$ and assume that $P(c. C[\vec{s},\vec{t}])$ holds for some predicate $P$ in $\pimp$. This means that $P(C[\vec{s},\vec{t}])$ holds. In order to preserve the soundness of the action law, $P(f(\vec{s}))$ should also hold, but this is impossible since $f$ is distinctive. One possible way of ensuring the soundness of the action law in the presence of implicit predicates is to stipulate some syntactic consistency requirements on the language specification. One sufficient requirement would be that if predicate rule $\frac{H'}{P(C[\vec{z},\vec{y}])}$ is derivable, then the system should contain a predicate rule $\frac{H''}{P(f[\vec{z}])}$ with $H'' \subseteq H'$. This is enough to guarantee that if the right-hand side of the action law satisfies $P$ then so does the left-hand side. 

\section{Soundness and completeness}
\label{sec:completeness}

Let us summarize our results so far. By Theorem~\ref{thm:s&c-smooth-dist}, it follows that, for any {\preg} system $G \sqsupseteq {\FINTREEPRED}^{\partial}$, there is an axiomatization that is head normalizing for $T(\Sigma \cup \Sigma_{\FINTREEPRED}^{\partial})$, where
$\Sigma \subseteq \Sigma_G \setminus \Sigma_{\FINTREEPRED}^{\partial}$ is a collection of smooth and distinctive operations of $G$.
Also, as hinted in Section~\ref{sec:smooth} (Lemma~\ref{lm:distinctive-law-all}), there exists a sound algorithm for transforming general {\preg} operations to smooth and distinctive ones.

So, for any {\preg} system $G$, we can build a {\preg} system $G' \sqsupseteq G$ and an axiomatization $\ET_{G'}$ that is head normalizing for $T(\Sigma_{G'})$. This statement is formalized as follows:

\begin{theorem}
\label{thm:s&c-fin-proc}
Let $G$ be a {\preg} system. Then there exist $G'\sqsupseteq G$ and a finite axiom system $\ET_{G'}$ such that
\begin{enumerate}\itemsep1pt
\item $\ET_{G'}$ is sound for bisimilarity on $T(\Sigma_{G'})$, \label{snd}
\item $\ET_{G'}$ is head normalizing for $T(\Sigma_{G'})$,
\end{enumerate}
and moreover, $G'$ and $\ET_{G'}$ can be effectively constructed from $G$.
\end{theorem}

\begin{proof}
{The result follows immediately by Theorem~\ref{thm:s&c-smooth-dist} and
by the existence of an algorithm used for transforming general {\preg} to smooth and distinctive operations}.
\end{proof}

\begin{remark}
Theorem~\ref{thm:s&c-fin-proc} guarantees ground-completeness of the generated axiomatization for well-founded {\preg} specifications, that is, {\preg} specifications in which each process can only exhibit finite behaviour.
\end{remark}

Let us further recall an example given in~\cite{Aceto:1994:TSR:184662.184663}. Consider the constant $\omega$, specified by the rule $\omega \xrightarrow{a} \omega$. Obviously, the corresponding action law  $\omega = a.\omega$ will apply for an infinite number of times in the  normalization process.
So the last step in obtaining a complete axiomatization is to handle infinite behaviour.

Let $t$ and $t'$ be two processes with infinite behaviour (remark that the infinite behaviour is a consequence of performing actions for an infinite number of times, so the extension to predicates is not a cause for this issue). Since we are dealing with finitely branching processes, it is well known that if two process terms are bisimilar at each finite depth, then they are bisimilar. One way of formalizing this requirement is to use the well-known \emph{Approximation Induction Principle} (AIP) \cite{Baeten:1991:PA:103272,Bergstra:1986:VAB:16663.16664}.

Let us first consider the operations $\pi_n(\cdot)$, $n \in \mathbb{N}$, known as \emph{projection operations}. The purpose of these operations is to stop the evolution of processes after a certain number of steps. The AIP is given by the following conditional equation:
\[
{x = y} \textnormal{ if } \pi_n(x) = \pi_n(y) ~ (\forall n \in \mathbb{N}).
\]

We further adapt the idea in \cite{Aceto:1994:TSR:184662.184663} to our context, and model the infinite family of projection operations $\pi_n(\cdot)$, $n \in \mathbb{N}$, by a binary operation $\cdot / \cdot$ defined as follows:

\begin{center}

\begin{tabular}{c@{~~~}c}
$\dfrac{x \xrightarrow{a} x' ~~ h \xrightarrow{c} h'}{x/h \xrightarrow{a} x'/h'} ~ (rl_{\rlaipa})$
&
$\dfrac{Px}{P(x/h)} ~ (rl_{\rlaipp})$ 
\end{tabular}
\end{center}
where $c$ is an arbitrary action. Note that $\cdot / \cdot$ is a smooth and distinctive operation.

The role of variable $h$ is to ``control" the evolution of a process, \textit{i.e.}, to stop the process in performing actions, after a given number of steps. Variable $h$ (the ``hourglass" in \cite{Aceto:1994:TSR:184662.184663}) will always be instantiated with terms of the shape $c^n$, inductively defined as: $c^0 = \delta$, $c^{n+1} = c.c^{n}$.

Let $G = (\Sigma_G, \setrulesG)$ be a {\preg} system. We use the notation $G_{/}$ to refer to the {\preg} system $(\Sigma_G \cup \{\cdot/\cdot\}, \setrulesG \cup \{(rl_{\rlaipa}), (rl_{\rlaipp})\})$ -- the extension of $G$ with $\cdot / \cdot$~. Moreover, we use the notation $\ET_{\it{AIP}}$ to refer to the axioms for the smooth and distinctive operation $\cdot / \cdot$, derived as in Section~\ref{sec:axiom-amoothAndSist} -- Definition~\ref{def:axioms}.

We reformulate AIP according to the new operation $\cdot/\cdot$~:
\[
x = y \textnormal{ if } x/c^n = y/c^n ~ (\forall n \in \mathbb{N})
\]

\begin{lemma}
\label{lm:aip}
AIP is sound for bisimilarity on $T(\Sigma_{\FINTREEPRED_{/}})$.
\end{lemma}

In what follows we provide the final ingredients for proving the existence of a ground-complete axiomatization for bisimilarity on {\preg} systems. As previously stated, this is achieved by reducing completeness to proving equality in {\FINTREEPRED}. So, based on AIP, it would suffice to show that for any closed process term $t$ and natural number $n$, there exists an {\FINTREEPRED} term equivalent to $t$ at moment $n$ in time:

\begin{lemma}
\label{lm:aip-ftp}
Consider $G$ a {\preg} system. Then there exist $G' \sqsupseteq G_{/}$ and $\ET_{G'}$ with the property: $\forall t \in T(\Sigma_{G'}), \forall n \in \mathbb{N}, \exists t' \in T(\Sigma_{\FINTREEPRED})$ s.t. $\ET_{G'} \vdash t/c^n = t'$.
\end{lemma}

At this point we can prove the existence of a sound and ground-complete axiomatization for bisimilarity on general $\preg$ systems:

\begin{theorem}[Soundness and Completeness]
\label{thm-s&c}
Consider $G$ a {\preg} system. Then there exist $G' \sqsupseteq G_{/}$ and $\ET_{G'}$ a finite axiom system, such that $\ET_{G'} \cup \ET_{\it{AIP}}$ is sound and complete for bisimilarity on $T(\Sigma_{G'})$.
\end{theorem}

\section{Motivation for handling predicates as first-class notions}
\label{sec:rationale}

In the literature on the theory of rule formats for Structural
Operational Semantics (especially, the work devoted to congruence
formats for various notions of bisimilarity), predicates are often
neglected at first and are only added to the considerations at a later
stage. The reason is that one can encode predicates quite easily by
means of transition relations. One can find a number of such encodings
in the literature---see, for instance,~\cite{assoc,Verhoef95}. In each
of these encodings, a predicate $P$ is represented as a transition
relation $\xrightarrow{P}$ (assuming that $P$ is a fresh action label)
with a fixed constant symbol as target. Using this translation, one can
axiomatize bisimilarity over {\preg} language specifications by first
converting them into ``equivalent'' standard GSOS systems, and then
applying the algorithm from~\cite{Aceto:1994:TSR:184662.184663} to
obtain a finite axiomatization of bisimilarity over the resulting GSOS
system.

In light of this approach, it is natural to wonder whether it is
worthwhile to develop an algorithm to axiomatize {\preg} language
specifications directly.  One possible answer, which has been
presented several times in the literature~\cite{Verhoef95}, is that
often one does not want to encode a language specification with
predicates using one with transitions only. Sometimes, specifications
using predicates are the most natural ones to write, and one should
not force a language designer to code predicates using
transitions. (However, one can write a tool to perform the translation
of predicates into transitions, which can therefore be carried out
transparently to the user/language designer.) Also, developing an
algorithm to axiomatize GSOS language specifications with predicates
directly yields insight into the difficulties that result from the
first-class use of, and the interplay among, various types of
predicates, as far as axiomatizability problems are concerned. These
issues would be hidden by encoding predicates as
transitions. Moreover, the algorithm resulting from the encoding would
generate axioms involving predicate-prefixing operators, which are
somewhat unintuitive.

Naturalness is, however, often in the eye of the beholder. Therefore,
we now provide a more technical reason why it may be worthwhile to
develop techniques that apply to GSOS language specifications with
predicates as first-class notions, such as the {\preg} ones. Indeed,
we now show how, using predicates, one can convert any standard GSOS
language specification $G$ into an equivalent {\em positive} one with
predicates $G^+$.

Given a GSOS language $G$, the system $G^+$ will have the same
signature and the same set of actions as $G$, but uses predicates
$\text{cannot}(a)$ for each action $a$. The idea is simply that ``$x
\:\text{cannot}(a)$'' is the predicate formula that expresses that ``$x$
does not afford an $a$-labelled transition''. The translation works as
follows.
\begin{enumerate}
\item Each rule in $G$ is also a rule in $G^+$, but one replaces each
  negative premise in each rule with its corresponding positive
  predicate premise. This means that $x \notranz{a}$ becomes $x
\:\text{cannot}(a)$.  
\item One adds to $G^+$ rules defining the predicates
  $\text{cannot}(a)$, for each action $a$. This is done in such a way
  that $p\:\text{cannot}(a)$ holds in $G^+$ exactly when $p\notranz{a}$
  in $G$, for each closed term $p$ and action $a$. More precisely, we
  proceed as follows.  
  \begin{enumerate}
  \item For each constant symbol $f$ and action $a$, add the
  rule 
  \[
  \frac{}{f \:\text{cannot}(a)}
  \]
  whenever there is no transition rule in $G$ with $f$ as principal
  operation and with an $a$-labelled transition as its consequent.
  \item  For each operation $f$ with arity
  at least one and action $a$, let $R(f,a)$ be the set of rules in $G$ that
  have $f$ as principal operation and an $a$-labelled transition as consequent. We want to
  add rules for the predicate $\text{cannot}(a)$ to $G^+$ that allow us to prove
  the predicate formula $f(p_1,\ldots,p_l) \:\text{cannot}(a)$ exactly when
  $f(p_1,\ldots,p_l)$ does not afford an $a$-labelled transition in $G$. This
  occurs if, for each rule in $R(f,a)$, there is some premise that
  is not satisfied when the arguments of $f$ are $p_1,\ldots,p_l$. To
  formalize this idea, let $H(R(f,a))$ be the collection of premises of
  rules in $R(f,a)$. We say that a choice function is a function $\phi:
  R(f,a) \rightarrow H(R(f,a))$ that maps each rule in $R(f,a)$ to one of its
  premises. Let 
  \begin{eqnarray*}
  \text{neg}(x \xrightarrow{a} x') & = & x \:\text{cannot}(a) \quad \text{and} \\
  \text{neg}(x \notranz{a}) & = & x \xrightarrow{a} x' ,  \quad \text{for some } x'.  \end{eqnarray*}
  Then, for each choice function $\phi$, we add
  to $G^+$ a predicate rule of the form 
  \[
  \frac{\{\text{neg}(\phi(\xi)) \mid \xi\in
  R(f,a)\}}{f(x_1,\ldots,x_l) \:\text{cannot}(a)} ,
  \]
  where the targets of the positive
  transition formulae in the premises are chosen to be all different.
  \end{enumerate}
\end{enumerate}
The above construction ensures the validity of the following lemma. 
\begin{lemma}
\label{lm:GtoGplus}
For each closed term $p$ and action $a$,
\begin{enumerate}
\item $p \xrightarrow{a} p'$ in $G$ if, and only if, $p \xrightarrow{a} p'$ in $G^+$;
\item $p \:\text{cannot}(a)$ in $G^+$ if, and only if, $p \notranz{a}$
in $G^+$ (and therefore in $G$).
\end{enumerate}
\end{lemma}
This means that two closed terms are bisimilar in $G$ if, and only if,
they are bisimilar in $G^+$. Moreover, two closed terms are bisimilar
in $G^+$ iff they are bisimilar when we only consider the transitions
(and not the predicates $\text{cannot}(a)$).

The language $G^+$ modulo bisimilarity can be axiomatized using our
algorithm without the need for the exponentially many
restriction operators. The conversion to positive GSOS with predicates
discussed above does incur in an exponential blow-up in the number of
rules, but it gives an alternative way of generating ground-complete
axiomatizations for standard GSOS languages to the one proposed
in~\cite{Aceto:1994:TSR:184662.184663}. In general, it is useful to
have several approaches in one's toolbox, since one may choose the one
that is ``less expensive'' for the specific task at hand.  Moreover,
using positive GSOS operations, one can also try to extend the methods
from the full version of the paper~\cite{DBLP:conf/concur/Aceto94}
(see Section 7.1 in the technical report available at
{\footnotesize{\url{http://www.ru.is/~luca/PAPERS/cs011994.ps}}}) to optimize these
axiomatizations. We are currently working on applying such methods to
positive {\preg} systems with universal as well as existential
predicates, and on extending our tool~ \cite{pregax-calco-tools2011}
accordingly.

It is worth noting that the predicates $\text{cannot}(a)$ are not implicit, therefore the restrictions presented at the end of Section~\ref{sec:axiom-amoothAndSist} need not to be imposed.

\section{Conclusions and future work}
\label{sec:conclusions}

In this paper we have introduced the {\preg} rule format, a natural extension of GSOS with arbitrary predicates. Moreover, we have provided a procedure (similar to the one in~\cite{Aceto:1994:TSR:184662.184663}) for deriving sound and ground-complete axiomatizations for bisimilarity of systems that match this format. In the current approach, explicit predicates are handled by considering constants witnessing their satisfiability as summands in tree expressions. Consequently, there is no explicit predicate $P$ satisfied by a term of shape $\Sigma_{i\in I}a_i.t_i$.

The procedure introduced in this paper has also enabled the implementation of a tool \cite{pregax-calco-tools2011} that can be used to automatically reason on bisimilarity of systems specified as terms built over operations defined by {\preg} rules.

Several possible extensions are left as future work. It would be worth investigating the properties of positive {\preg} languages. By allowing only positive premises we eliminate the need of the restriction operators (\onedag) during the axiomatization process. This would enable us to deal with more general predicates over trees, such as those that may be satisfied by terms of the form $a.t$ where $a$ ranges over some subset of the collection of actions.

Another direction for future research is that of understanding the presented work from a coalgebraic perspective. The extensions from \cite{Aceto:1994:TSR:184662.184663} to the present paper, might be thought as an extension from coalgebras for a functor $\mathscr{P}(\calA \times \textit{Id})$ to a functor $\mathscr{P}(\calP) \times \mathscr{P}(\calA \times \textit{Id})$ where $\mathscr{P}$ is the powerset functor, $\calA$ is the set of actions and $\calP$ is the set of predicates. Also the language {\FINTREEPRED} coincides, apart from the recursion operator, with the one that would be obtained for the functor $\mathscr{P}(\calP) \times \mathscr{P}(\calA \times \textit{Id})$ in the context of Kripke polynomial coalgebras \cite{DBLP:conf/lics/BonsangueRS09}.

Finally, we plan to extend our axiomatization theory in order to reason on the bisimilarity of guarded recursively defined terms, following the line presented in \cite{DBLP:conf/concur/Aceto94}.

\paragraph{Acknowledgments.} The authors are grateful for the useful comments and suggestions from Alexandra Silva and three anonymous reviewers.

\bibliographystyle{eptcs}
\bibliography{formal}

\end{document}